\documentclass[12pt,a4paper]{article}
\usepackage{bbm}
\usepackage{mathrsfs}
\usepackage{latexsym}
\usepackage{lmodern}
\usepackage{paralist}
\usepackage{microtype} % Prettifies pdf output. Uncomment if you have trouble with this. Will also reduce page count.

\usepackage{xpunctuate}
\usepackage[all,british]{foreign} % Typesets abbreviations like e.g.

\usepackage{amsmath,amsfonts,amssymb,amstext,amsthm} % AMS packages
\usepackage{thmtools, thm-restate}
\usepackage{mathtools}

\usepackage[full,small]{complexity}

\usepackage[notref,color,final]{showkeys} % Use option 'final' to remove

\usepackage{xargs}
\usepackage{calc}

\PassOptionsToPackage{dvipsnames,usenames,table}{xcolor}
\usepackage{tikz}
\usepackage{tkz-graph}
\usetikzlibrary{arrows}
\usetikzlibrary{decorations.pathmorphing}
\usetikzlibrary{calc}

\usepackage{booktabs} % Better tables
\usepackage{longtable}
\usepackage{float}
\usepackage{wrapfig}
\usepackage{floatflt}
\usepackage{framed}
\usepackage{ctable}

\def\figscale{.8}
\def\smallfigscale{.6}

\usepackage{enumerate}
\usepackage{enumitem}
\usepackage{xspace}
\usepackage{ifthen}
\usepackage{fixltx2e}
\usepackage{url}
\usepackage{scrtime}
\usepackage{tikz}

\usepackage[ruled,longend,vlined]{algorithm2e}

\usepackage{xkvltxp}
\usepackage[footnote,draft,english,silent,nomargin]{fixme}
\newcommand{\todo}[1]{\fxfatal{\color{red}#1}}

\renewcommand{\cal}{\mathcal}

\newcommand{\mc}{\mathcal}

\newcommand{\Problem}[1]{\textsc{#1}\xspace}

\newcommand\overbar[2][3]{{}\mkern#1mu\overline{\mkern-#1mu#2}}
\def\comp#1{\overbar #1} % Complement
\newcommand\restr[2]{{% Restriction of functions, set families
  \left.\kern-\nulldelimiterspace % automatically resize the bar with \right
  #1 % the function
  \vphantom{\big|} % pretend it's a little taller at normal size
  \right|_{#2} % this is the delimiter
  }}

\def\subgraph{\subseteq}

\def\sminor^#1{\preccurlyeq_{{\mathit{m}}}^{#1}\!} % Shallow minor relation
\def\stminor^#1{\preccurlyeq_{{\mathit{t}}}^{#1}\!} % Shallow minor relation

\def\lexprod{\mathbin{{}\bullet{}}}

\def\grad_#1{\nabla\!_#1}
\def\topgrad_#1{\widetilde \nabla\!_{#1}}

\renewcommand{\le}{\leqslant}
\renewcommand{\leq}{\leqslant}

\renewcommand{\geq}{\geqslant}
\renewcommand{\emptyset}{\varnothing} % nicer empty set

\newcommand{\widthm}[1]{\ensuremath{\mathop\mathbf{#1}}\xspace}

\newcommand{\width}{\widthm{width}}

\newcommand{\dist}{\ensuremath{\text{dist}}}

\def\any{\mathord{\color{black!55}\circ}}%
\DeclarePairedDelimiter\ceil{\lceil}{\rceil}

\DeclareMathOperator*{\argmax}{arg\,max}

\def\nbeq_#1{\simeq^{X}_{#1}}
\def\snbeq_#1{\cong^{X}_{#1}}
\def\rsigeq_#1{\simeq^{X}_{\Sig_{\leq #1}}\!}
\def\prsigeq_#1{\simeq^{X}_{\hat \Sig_{\leq #1}}\!}
\def\sigeq_#1{\simeq^{X}_{#1}\!}

\theoremstyle{plain}
\newtheorem{lemma}{Lemma}

\newtheorem{theorem}{Theorem}
\newtheorem{corollary}{Corollary}
\newtheorem{observation}{Observation}
\newtheorem{proposition}{Proposition}

\newtheoremstyle{case}
  {\topsep}   % ABOVESPACE
  {\topsep}   % BELOWSPACE
  {}  % BODYFONT
  {\parindent}       % INDENT (empty value is the same as 0pt)
  {\bfseries} % HEADFONT
  {.}         % HEADPUNCT
  {5pt plus 1pt minus 1pt} % HEADSPACE
  {#1 #2: {\normalfont #3}}          % CUSTOM-HEAD-SPEC

\theoremstyle{case}

\numberwithin{subcase}{case}
\makeatletter%Make sure case-counter is reset in every environment
\@addtoreset{case}{lemma}
\@addtoreset{case}{theorem}
\@addtoreset{case}{corollary}
\makeatother

\theoremstyle{definition}

\newtheorem{definition}{Definition}

\renewcommand*\etal{\xperiodafter{\emph{et~al}}}

\setlist[1]{labelindent=\parindent,leftmargin=*} 
\setlist{itemsep=0pt}
\setitemize[1]{label={\small\textbullet}}

\renewenvironmentx{leftbar}[2][1=0.5pt, 2=5pt]{% 
  \def\FrameCommand{\vrule width #1 \hspace{#2}}\MakeFramed {\advance\hsize-\width \FrameRestore}}%
{\endMakeFramed}

\makeatletter
\newcommand{\pgfsize}[2]{ % #1 = width, #2 = height
 \pgfextractx{\@tempdima}{\pgfpointdiff{\pgfpointanchor{current bounding box}{south west}}
 {\pgfpointanchor{current bounding box}{north east}}}
 \global#1=\@tempdima
 \pgfextracty{\@tempdima}{\pgfpointdiff{\pgfpointanchor{current bounding box}{south west}}
 {\pgfpointanchor{current bounding box}{north east}}}
 \global#2=\@tempdima
}
\makeatother

\newlength{\wleft}  \newlength{\wright}

\makeatletter
\def\namedlabel#1#2{\begingroup
   \def\@currentlabel{#1}%
   \label{#2}\endgroup
}
\makeatother

\usepackage{hyphenat}
\def\Nesetril{Ne\v{s}et\v{r}il\xspace}

\def\Dvorak{Dvo\v{r}\'{a}k\xspace}

\def\Kral{Kr\'{a}l\xspace}

\def\Gajarsky{Gajarsk{\'y}\xspace}

\newlength{\templen}

\newcommand{\NOdM}{\Nesetril and Ossona de Mendez}
\newcommand{\wcol}{\mathrm{wcol}}

\newcommand{\Wreach}{\mathrm{WReach}}

\usepackage{authblk}

\title{Characterising Bounded Expansion by Neighbourhood Complexity}
\author{\hspace*{.5em}Felix Reidl\hspace*{2em}Fernando S\'anchez Villaamil\\Konstantinos Stavropoulos}
\date{\vspace*{-2em}}

\affil{  
RWTH Aachen University 

\texttt{\{reidl,fernando.sanchez,stavropoulos\}@cs.rwth-aachen.de}.
}

\begin{document}

\maketitle 

\begin{abstract}
We show that a graph class $\cal G$ has \emph{bounded expansion} if and only
if it has bounded \emph{$r$-neighbourhood complexity}, \ie for any vertex set
$X$ of any subgraph~$H$ of $G\in\cal G$, the number of subsets of $X$ which
are exact $r$-neighbourhoods of vertices of $H$ on $X$ is linear in the size
of $X$. This is established by bounding the $r$-neighbourhood complexity of a
graph in terms of both its \emph{$r$-centred colouring number} and its
\emph{weak $r$-colouring number}, which provide known characterisations to the
property of bounded expansion.
\end{abstract}

\section{Introduction}\label{sec:intro}

\noindent
Graph classes of \emph{bounded expansion} (and their further generalisation, 
nowhere dense classes) have been introduced by  \NOdM~\cite{BndExpI,BndExpII,Sparsity}
as a general model of \emph{structurally sparse} graph classes. They include
and generalise many other natural sparse graph classes, among them all
classes of bounded degree, classes of bounded genus, and classes defined by
excluded (topological) minors. Nowhere dense classes even include classes that
locally exclude a minor, which in turn generalises graphs with locally bounded
treewidth.

The appeal of this notion and its applications stems from
the fact that bounded expansion has turned out to be a very robust property of
graph classes with various seemingly unrelated characterisations (see
\cite{GKS13,Sparsity}). These include characterisations through
the density of shallow minors~\cite{BndExpI}, \emph{quasi-wideness}~\cite{QuasiWide}
\emph{low treedepth colourings}~\cite{BndExpI}, and \emph{generalised colouring numbers}~\cite{WColBndExp}.
The latter two are particularly relevant towards algorithmic applications,
as we will discuss in the sequel. Furthermore, there is good evidence that 
real-world graphs (often dubbed `complex networks') might exhibit this notion
of structural sparseness~\cite{NetworksBndExp,FelixThesis}, whereas stricter
notions (planar, bounded degree, excluded (topological) minors, \etc) do not apply.

It seems unlikely that bounded-expansion and nowhere dense classes admit global
Robertson-Seymour style decompositions as they are available for classes excluding
a fixed minor~\cite{GraphMinorsXVI}, a topological minor~\cite{GroheMarxDecomp}, 
an immersion~\cite{ImmersionDecomp}, or an odd minor~\cite{OddMinorDecomp}. However,
\Nesetril and Ossona de Mendez showed~\cite{BndExpII} that bounded-expansion and nowhere dense
classes admit a `local' decomposition, a so-called \emph{low $r$-treedepth colouring}, in the following sense: 
for every integer~$r$, every graph~$G$ from a bounded expansion (nowhere dense) class can be
coloured with~$f(r)$ (respectively~$O(n^{o(1)})$)
colours such that every union of~$p < r$ colour classes induces a graph of
treedepth at most~$p$. We denote by $\chi_r(G)$ the minimal number of
colours needed for a low $r$-treedepth colouring of $G$.
These types of colourings generalise the star-colouring
number~\cite{Sparsity} introduced by Fertin, Raspaud,
and Reed~\cite{StarChromaticNumber}. In that context, low $r$-treedepth
colourings are usually
called \emph{$r$-centred colourings}\footnote{
Depending on the way $r$-treedepth colourings are defined, $r$-centred colourings 
might appear in the literature as $r-1$-treedepth colourings, as for example in \cite{Sparsity}. For convenience, here we define them 
in a way so that the gap in the depth $r$ is alleviated.}
(the precise 
definition of which we defer to Section \ref{sec:preliminaries}). 
%The notion of choice to be used in this paper, equivalent to the $r$-treedepth colourings 
%as we defined them above

This `decomposition by colouring' has direct algorithmic implications. For
example, counting how often an~$h$-vertex graph appears in a host graph~$G$ as
a subgraph, induced subgraph or homomorphism is possible in linear
time~\cite{BndExpII} through the application of low $r$-centred ($r$-treedepth) colourings. A
more precise bound for the running time
of $O(|c(G)|^{2h} 6^h h^2 \cdot |G|)$ was shown by Demaine
\etal~\cite{NetworksBndExp} if an appropriate low treedepth
colouring $c$ is provided as input. 
Low $r$-centred ($r$-treedepth) colourings can be further used to check whether an existential 
first-order sentence is true~\cite{Sparsity} or to approximate the problems
\Problem{$\cal F$-Deletion} and \Problem{Induced-$\cal F$-Deletion}
to within a factor that only depends on the 
precise bounded expansion graph class $G$ belongs to
and the set~$\cal F$~\cite{FelixThesis}.

Another characterisation of bounded expansion is obtained via the
\emph{weak $r$-colouring numbers}, denoted by $\wcol_{r}(G)$. The name
`colouring number' reflects the fact that the weak $1$-colouring number is
sometimes also called the \emph{colouring number} of the graph, which
only differs to the \emph{degeneracy} of a graph by one.
Roughly, the weak colouring
number describes how well the vertices of a graph can be linearly ordered such
that for any vertex~$v$, the number of vertices that can reach~$v$ via short
paths that use higher-order vertices is bounded.
We postpone the precise definition of
weak~$r$-colouring numbers to Section \ref{sec:preliminaries}, but let us
emphasise their utility: Grohe, Kreutzer, and
Siebertz~\cite{FONowhereDense} used weak $r$-colouring numbers to prove the milestone result that
first-order formulas can be decided in almost linear time for nowhere-dense classes
(improving upon a result by \Dvorak, \Kral, and Thomas for bounded expansion
classes~\cite{FOBndExp} and the preceding work for smaller sparse
classes~\cite{FOLocalHMinorFree,FOHMinorFree,FOLocalBndTw,FOBndDegree}).

Our work here centres on a new characterisation, motivated by recent progress in
the area of kernelisation. This field, a subset of parametrised complexity theory,
formalises polynomial-time preprocessing of computationally hard problems. For an
introduction to kernelisation we refer the reader to the seminal work by Downey and
Fellows~\cite{DowneyFellows}.
\Gajarsky \etal~\cite{BndExpKernels} extended the meta-kernelisation framework initiated 
by Bodlaender \etal~\cite{Metakernels} for bounded-genus graphs to nowhere-dense
classes (notable intermediate results where previously obtained for
excluded-minor classes~\cite{BidimKernels} and classes excluding a topological
minor~\cite{HTopFreeKernels}). In a largely independent line of research,
Drange~\etal recently provided a kernel for \textsc{Dominating Set} on
nowhere-dense classes~\cite{kernelmillionauthors}.
Previous results showed kernels for planar
graphs~\cite{PlanarDSKernel}, bounded-genus graphs~\cite{Metakernels}, apex-
minor-free graphs~\cite{BidimKernels}, graphs excluding a
minor~\cite{DomsetKernelHMinorFree} and graphs excluding a topological
minor~\cite{DomsetKernelHTopFree}.

A feature exploited heavily in the above kernelisation results 
for bounded expansion classes is that for any
graph~$G$ from such a class, every subset~$X \subseteq G$ has the property
that the number of ways vertices from~$V(G) \setminus X$ connect to~$X$ is linear
in the size of~$X$. Formally, we have that
\[
  |\{ N(v) \cap X \}_{v \in V(G)}| \leq c \cdot |X|
\]
where~$c$ only depends on the graph class from which~$G$ was drawn. One
wonders whether this property of bounded expansion classes can be turned into a
characterisation. It is, however, missing one important ingredient present in
all known notions related to bounded expansion: a notion of \emph{depth} via an
appropriate distance-parameter. This brings us to the central notion of our work: 
If we denote by~$N^r[\any]$ the closed
$r$-neighbourhood around a vertex, we define the \emph{$r$-neighbourhood
complexity} as
\[
  \nu_r(G) := 
    \max_{H \subseteq G, \emptyset \neq X \subseteq V(H)}
     \frac{|\{ N^r[v] \cap X\}_{v \in H} |}{|X|}.
\]
That is, the value~$\nu_r$ tells in how many different ways vertices can
be joined to a vertex set~$X$ via paths of length at most~$r$. Note that we define
the value over all possible subgraphs: otherwise uniform dense graphs
(\eg complete graphs) would yield very low values\footnote{While this might be an
interesting measure in and of itself, in this work we want to develop a
measure for sparse graph classes and therefore choose the above definition.}.

The main result of this paper is the following characterisation of bounded
expansion through neighbourhood complexity. We say that a graph class~$\mc G$
has \emph{bounded neighbourhood complexity} if there exists a function~$f$
such that for every~$r$ it holds that $\nu_r(\mc G) \leq f(r)$.

\begin{theorem}\label{thm:bndexp-equals-bndnc}
  A graph class $\cal G$ has bounded expansion if and only if it has
  bounded neighbourhood complexity.
\end{theorem}

\noindent
Specifically, we show that the following relations between the $r$-neighbourhood
complexity $\nu_r$, the $r$-centred colouring number $\chi_r$, and the weak $r$-colouring number $\wcol_r$
of a graph.

\begin{restatable}{theorem}{nbcentred}\label{thm:nb-bound-centred} 
  For all graphs~$G$ and all non-negative integers~$r$ it holds that
  \[
    \nu_r(G) \leq (r+1)2^{\chi_{2r+2}(G)^{r+2}}.
  \]
\end{restatable} 

\begin{restatable}{theorem}{nbweak}\label{thm:nb-bound-weak} 
  For every graph~$G$ and all non-negative integers~$r$ it holds that
  \[
    \nu_r(G) \leq \frac{1}{2}(2r+2)^{\wcol_{2r}(G)}\wcol_{2r}(G)+1.
  \]
\end{restatable}

\noindent
The characterisation of bounded expansion through
generalised colouring numbers in \cite{WColBndExp} was provided by
relating $r$-centred colourings to generalised colouring numbers. We
believe that this interaction of the two notions is also highlighted in this
paper, in the sense that when one can use one of the two notions as a direct proof
tool, it might often be the case that the other might also serve as a direct proof
tool, the most appropriate to be chosen depending on the occasion. 
As we believe it is also the case with neighbourhood complexity, 
it is still, as a consequence, useful to have access to a result through both parameters, 
since the general known bounds relating $r$-centred colourings
and generalised colouring numbers seem to be very loose and most probably not
optimal. For example, it is still unclear to our knowledge if one is always smaller than the
other. Moreover, bounds for both parameters are not in general known for all
kinds of specific graph glasses. It can then be the case that for different questions and different
graph classes, $r$-centred colourings are more appropriate than generalised
colouring numbers or vice versa.

\def\sig{\sigma} % Maybe sigma is not a good idea.
\def\Sig{\cal S} % Maybe sigma is not a good idea.

%   888888ba   888888ba  dP        8888ba.88ba  .d88888b  
%   88    `8b  88    `8b 88        88  `8b  `8b 88.    "' 
%  a88aaaa8P' a88aaaa8P' 88        88   88   88 `Y88888b. 
%   88         88   `8b. 88        88   88   88       `8b 
%   88         88     88 88        88   88   88 d8'   .8P 
%   dP         dP     dP 88888888P dP   dP   dP  Y88888P  
%                                                         
%                                                         
\section{Preliminaries}\label{sec:preliminaries}

\noindent
The main challenge is to prove that graphs from a graph class
of bounded expansion have low neighbourhood complexity. To this end, some 
definitions will be necessary to prove Theorems~\ref{thm:nb-bound-centred} 
and~\ref{thm:nb-bound-weak}.

\subsection{Graphs and Signatures}

\noindent
For an integer~$n$ we write~$[n] = \{1,\ldots,n\}$. All logaritms in this paper are of base $2$
and we only write $ \log x$ instead of $\log_2 x$. 
We only consider non-empty, finite and simple graphs. For a graph~$G$ we write~$V(G)$ and~$E(G)$
to denote vertices and edges of~$G$, respectively. We use the
notations~$|G| = |V(G)|$ and~$\|G\| = |E(G)|$. Following the notation
of Diestel~\cite{Diestel}, we denote an edge between two nodes $u,v \in V(G)$
by $uv$. In the following we will sometimes use the symbol~$\any$ to
denote an arbitrary vertex and it should be understood that each
occurence of~$\any$ can denote a different vertex. The statement
`there exist two edges~$v\any$,$w\any$' therefore means~`there exist
two edge~$va_1$,$va_2$' and we will prefer the former if $a_1,a_2$ are
not referenced later.

For a vertex~$v
\in V(G)$, we denote by $N^r_G(v) := \{ u \in V(G) \mid \dist_G(u,v) =r\}$ the
\emph{$r$-th neighbourhood around~$v$} for $r \geq 0$. Analogously, the
$r$-th \emph{closed} neighbourhood around~$v$ is defined as
$N^r_G[v] := \bigcup_{i = 0}^r N^i(v)$. In particular,
$N_G^0(v) = N_G^0[v] = \{v\}$.
We usually omit the subscript~$G$ if the context is clear.

A \emph{signature}~$\sig$ over a universe~$U$ is a sequence of 
elements~$(u_i)_{1 \leq i \leq \ell}, u_i \in U$ where $\ell$ is the
\emph{length} of the signature, also denoted by~$|\sig|$. Accordingly, an
\emph{$\ell$-signature} is simply a signature of length~$\ell$.
We use the notation~$\sig[i] := u_i$ to signify the $i$-th element
of~$\sig$. A signature is \emph{proper} if all its elements
are distinct. We assume that the elements of $U$ are ordered.
We assume an order on all signatures (say, lexicographic).
Thus for a set~$\Sig$ of signatures and a
function~$f\colon \Sig \to A$ for an arbitrary set~$A$, we employ
the notation~$(f(\sig))_{\sig \in \Sig}$ to obtain sequences over elements
of~$A$ derived from that ordering. For example, $(|\sig|)_{\sig \in \{\sig_a,\sig_b,\sig_c\}}$
is shortand for the sequence~$(|\sig_a|,|\sig_b|,|\sig_c|)$ if~$\sig_a \leq \sig_b \leq \sig_c$
according to our (arbitrary) total order.

For a path~$P = x_1\ldots x_\ell$ we write~$P[x_i,x_j] = x_i \ldots x_j$ to
denote the subpath of~$P$ starting at~$x_i$ and ending at~$x_j$. As such,
we treat paths as ordered. Similarly, for an integer~$1 \leq i \leq |P|$
we denote by $P[i]$ the $i$-th vertex on the path and we call~$i$ the
\emph{index} of that vertex on~$P$. Hence, for non-empty paths, $P[1]$ is the
start and $P[|P|]$ the end of the path. If~$G$ is a graph coloured by~$c\colon
V(G) \to [\xi]$ for some~$\xi \in \mathbb{N}$ and~$P$ is a path in~$G$, then
we write $\sig_{P}$ to denote the $|P|$-signature over~$[\xi]$ with
$\sig_{P}[i] = c(P[i])$. For a fixed signature~$\sig$, we say that $P$ is a
\emph{$\sig$-path} if $\sig_{P} = \sig$.

For a fixed signature~$\sig$ over~$[\xi]$, we define the 
\emph{$\sig$-neighbourhood} of a vertex~$v$ in $G$ as
\[
  N^{\sig}(v) := \{ w \in V(G) \mid \exists vPw~\text{such that}~ \sig_{vPw} = \sig \}
\]
Note that $N^{\sig}(v) \subseteq N^{|\sig|}(v)$ and that~$N^{\sig}(v) = \emptyset$ whenever
$\sig[1] \neq c(v)$. We use the following extension to vertex sets~$X \in V(G)$
and sets of signatures~$\Sig$ over~$[\xi]$:
\begin{align*}
  N^{\Sig}(v) &:= \bigcup_{\sig \in \Sig} N^{\sig}(v) &&&
  N^{\sig}(X) &:= \bigcup_{v \in X} N^{\sig}(v) &&&
  N^{\Sig}(X) &:= \bigcup_{v \in X} \bigcup_{\sig \in \Sig} N^{\sig}(v)  
\intertext{%
Similarly, the \emph{$\sig$-in-neighbourhood} of a vertex~$v$ is defined 
\[
  N^{-\sig}(v) := \{ w \in V(G) \mid \exists wPv~\text{such that}~ \sig_{wPv} = \sig \}
\]
and we extend this notation to vertex and signature sets in the same manner as above:
} 
  N^{-\Sig}(v) \! &:= \!\bigcup_{\sig \in \Sig} \! N^{-\sig}(v) &&&
  N^{-\sig}(X) \! &:= \!\bigcup_{v \in X} \! N^{-\sig}(v) &&&
  N^{-\Sig}(X) \! &:= \!\bigcup_{v \in X} \! \bigcup_{\sig \in \Sig} \! N^{-\sig}(v)  
\end{align*}
The following basic fact about $\sigma$-neighbourhoods for proper signatures~$\sigma$
is easy to verify.

\begin{observation}\label{obs:centre-intersect}
  Let $u,v \in V(G)$ be distinct vertices and $uP\any$, $vP\any$
  be two $\sig$-paths for some proper signature $\sig$. 
  Then for any $x \in u\sig\any \cap v\sig\any$
  it holds that $x$ has the same index on both $u\sig\any$ and $v\sig\any$
  and that $x$s colour appears exactly once in $uP\any \cup vP\any$.
\end{observation}

\noindent
Finally the \emph{lexicographic product} $G_1 \lexprod G_2$ is the
graph with vertices $V(G_1) \times V(G_2)$, where two nodes $(u,x)$
and $(v,y)$ are connected by an edge iff either
\begin{inparaenum}
\item[a)] $uv \in E(G_1)$ or
\item[b)] $u = v ~\text{and}~ xy \in E(G_2)$.
\end{inparaenum}

\subsection{Grad and Expansion}

\noindent
The property of \emph{bounded expansion} was introduced by
\Nesetril and Ossona de Mendez using the notion of shallow
minors~\cite{BndExpI,BndExpII}: the basic idea is to exclude
different minors depending on how `local' the contracted portions of the graph
is. Building on \Dvorak's work~\cite{DvorakThesis}, \Nesetril, Ossona de Mendez,
and Wood later introduced an equivalent definition via shallow
\emph{topological} minors~\cite{BndExpClasses}. This seem surprising at first, since graphs defined
via (unrestricted) forbidden minors are vastly different objects than graphs
defined via forbidden \emph{topological} minors. We will only introduce the
topological variant here.

\begin{definition}[Topological minor embedding]\label{def:shallow-topminor}
  A \emph{topological minor embedding} of a graph~$H$ into a graph~$G$
  is a pair of functions~$\phi_V\colon V(H) \to V(G)$, $\phi_E\colon E(H) \to 2^{V(H)}$ where $\phi_V$ is injective
  and for every~$uv \in H$ we have that
  \begin{enumerate}
    \item $\phi_E(uv)$ is a path in~$G$ with endpoints~$\phi_V(u), \phi_V(v)$ and
    \item for every~$u'v' \in H$ with~$u'v'\neq uv$ the two paths
          $\phi_E(uv)$, $\phi_E(u'v')$ are internally vertex-disjoint. 
  \end{enumerate}
\end{definition}

\noindent
We define the \emph{depth} of the topological minor embedding~$\phi_V,\phi_E$
as the half-integer $(\max_{uv \in H} |\phi_E(uv)| - 1)/2$, \ie an embedding
of depth~$r$ will map the edges of~$H$ onto paths in~$G$ of length at most~$2r+1$.

Accordingly, if~$H$ has a topological minor embedding of depth~$r$ into~$G$
we say that $H$ is a an \emph{$r$-shallow topological minor} of~$G$ and
write~$H \stminor^r G$. Note that this relationship is monotone in the sense
that an $r$-shallow topological minor of~$G$ is also an~$r+1$-shallow topological minor of~$G$.

\begin{definition}[Grad and bounded expansion]
  For a graph~$G$ and an integer~$r \geq 0$, we define the \emph{topologically greatest
    reduced average density (top-grad) at depth~$r$} as
  \[
    \topgrad_r(G) = \max_{H \stminor^r G}  \frac{\|H\|}{|H|}.
  \]
  We extend this notation to graph classes as $\topgrad_r(\mc G) =
  \sup_{G \in \mc G} \topgrad_r(G)$.  A graph class $\mc G$ then has
  \emph{bounded expansion} if there exists a function $f \colon \mathbb{N} \to
  \mathbb{R}$ such that for all~$r$ we have that $\topgrad_r(\mc G) \leq f(r)$.
\end{definition}

\subsection{$r$-Centred Colourings and Weak $r$-Colouring Number}

\noindent
Equivalent definitions for classes of bounded expansion are related
to the $r$-centred colouring number and the weak $r$-colouring number of graphs.

\begin{definition}[$r$-centred colourings]
  An $r$-centred colouring of a graph $G$ is a vertex colouring such
  that, for any (induced) connected subgraph $H$, either some colour
  $c(H)$ colours exactly one node (a \emph{centre}) in $H$ or $H$ gets
  at least $r$ colours.
\end{definition}

\noindent
The minimum number of colours of an $r$-centred colouring of $G$ is denoted by $\chi_r(G)$. 
Let us see the characterisation of bounded expansion via $\chi_r$.

\begin{proposition}[\Nesetril, Ossona de Mendez~\cite{BndExpI}]\label{prop:low-td-colourings}
  Let~$\mc G$ be a graph class of bounded expansion. Then there
  exists a function~$f_c$ such that for every~$r \in \mathbb{N}$ and
  every~$G \in \mc G$ it holds that $\chi_r(G) \leq f_c(r)$.
\end{proposition}

\noindent
Let $\Pi(G)$ be the set of linear orders on $V(G)$ and let $\preceq\;\in
\Pi(G)$. We represent $\preceq$ as an injective function $L\colon V(G) \rightarrow \mathbb{N}$ with the property that $v\preceq w$ if and only if $L(v) \le L(w)$.

A vertex $u$ is \emph{weakly $r$-reachable} from $v$ with respect to the order
$L$, if there is a path $P$ of length at most $r$ from $v$ to $u$ such
that $L(u) \le L(w)$ for all $w\in V(P)$.  Let $\Wreach_r[G,L, v]$ be the set
of vertices that are weakly $r$-reachable from $v$ with respect to~$L$.
The \emph{weak $r$-colouring number} $\wcol_r(G)$ is now defined as
\[
  \wcol_r(G)=\min_{L\in\Pi(G)}\max_{v\in V(G)}|\Wreach_r[G,L, v]|.
\]
For a set of vertices $X\subseteq V(G)$, we let 
\[
  \Wreach_{r}[G,L, X]=\bigcup\limits_{v\in X} \Wreach_{r}[G,L, v].
\]
Zhu \cite{WColBndExp} showed that a graph class has bounded
expansion if and only if the weak $r$-colouring number~$\wcol_r$
of every member is bounded by a function that only depends on~$r$.

\subsection{Neighbourhood Complexity}

\noindent
\begin{definition}[Neighbourhood complexity]
  For a graph~$G$ the \emph{$r$-neighbourhood complexity} is a function
  $\nu_r$ defined via
  \[
    \nu_r(G) := \max_{H \subgraph G, \emptyset\neq X \subseteq V(H)} \frac{|\{ N^r[v] \cap X \}_{v \in V(H)}|}{|X|}.
  \]
  We extend this definition to graph classes~$\mc G$ via 
  $\nu_r(\mc G) := \sup_{G \in \mc G} \nu_r(G)$.
\end{definition}

\noindent
Alternatively, we can define the neighbourhood complexity
via the index of an equivalence relation. This turns out to be 
a useful perspective in the subsequent proofs.
For~$r \in \mathbb N$ and~$X \subseteq V(G)$, we define the \emph{$(X,r)$-twin
equivalence} over~$V(G)$ as
\[
  u \simeq^{G,X}_r v \iff N^r[u] \cap X = N^r[v] \cap X
\]
which gives rise to the alternative definition
\[
  \nu_r(G) = \max_{H \subgraph G, \emptyset\neq X \subseteq V(G)}
            \frac{|V(H) / {\simeq^{H,X}_{r}} |}{|X|}.
\]
We will usually 
fix a graph in the following and hence omit the superscript~$G$
of this relation. Recall that we say that a graph class~$\mc G$ has \emph{bounded neighbourhood complexity}
if there exists a function~$f$ such that for every~$r$ it holds
that $\nu_r(\mc G) \leq f(r)$.

%   a88888b. 888888ba  d888888P  888888ba  888888ba  
%  d8'   `88 88    `8b    88     88    `8b 88    `8b 
%  88        88     88    88    a88aaaa8P' 88     88 
%  88        88     88    88     88   `8b. 88     88 
%  Y8.   .88 88     88    88     88     88 88    .8P 
%   Y88888P' dP     dP    dP     dP     dP 8888888P  
%                                                    
%        
\section{Neighbourhood Complexity and $r$-Centred Colourings}

\noindent
This section is dedicated to proving the following relation between the
$r$-neighbourhood complexity and the $(2r+2)$-centred colouring number of a graph.

\nbcentred*
% \begin{theorem}\label{thm:nb-bound-centred} 
%   For all graphs~$G$ and all integer~$r$ it holds that
%   \[
%     \nu_r(G) \leq 2^{\chi_{2r+2}(G)^{r+2}}.
%   \]
% \end{theorem}

\noindent 
For the remainder of this section, we fix a graph~$G$, a subset of vertices
$\emptyset\neq X\subseteq V(G)$, an integer~$r$ and a $(2r+2)$-centred colouring~$c \colon V(G)
\to [\xi]$ where $\xi = \chi_{2r+2}(G)$. We will assume that~$G$ and~$X$
are chosen such that
$|V(G) / {\simeq^{G,X}_r} | = \nu_r(G) \cdot |X|$. For readability we 
will drop the superscript~$G$ from~$\simeq^{G,X}_r$ in the following.

In the following we introduce a sequence of equivalence relations over~$V(G)$
and prove that they successively refine $\simeq^X_r$. To that end,
define~$\Sig_{\leq r}$  to be the set of all signatures over~$[\xi]$ of
length at most~$r$.
The subsequent lemmas will elucidate the connection between
centred colourings and proper signatures.

\begin{lemma}\label{lem:sigma-nb}
  For any \emph{proper} signature $\sig \in \Sig_{\leq r}$ and any
  vertices $u,v \in V(G)$, either $N^\sig(u) \cap
  N^\sig(v) = \emptyset$ or $N^\sig(u) = N^\sig(v)$. 
\end{lemma}
\begin{proof}
  Assume there exists $x \in N^\sig(u) \cap N^\sig(v)$ but $N^\sig(u) \neq N^\sig(v)$.
  Without loss of generality, let $y \in N^\sig(v) \setminus N^\sig(u)$.

  \begin{center}
    \includegraphics[scale=\figscale]{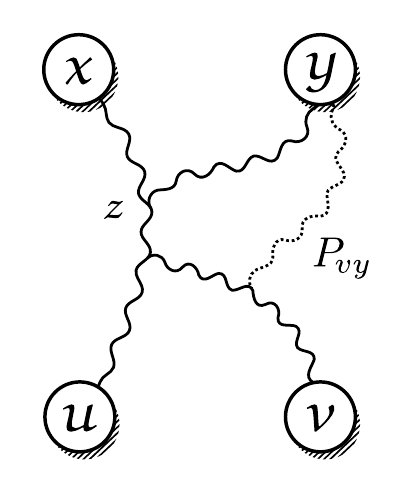}
  \end{center}

  \noindent
  Fix a $\sig$-path $P_{ux}$ and a $\sig$-path $P_{vx}$. Let $s \in P_{ux} \cap P_{vx}$
  be the first vertex in which both paths intersect (since both paths end in
  $x$, such a vertex   must exist). Further, fix a $\sig$-path $P_{vy}$. Now if
  $P_{vy} \cap P_{ux}$ is non-empty,   then $y$ is $\sig$-reachable from $u$: by
  Observation~\ref{obs:centre-intersect}, there would   be a vertex $z \in P_{vy}
  \cap P_{ux}$ that has the same index on both paths. Since $\sig$ is proper, the
  subpath of $P_{vy}[z,y]$ cannot share a vertex with $P_{ux}[u,z]$, thus we
  can construct a $\sig$-path by first taking the subpath~$P_{ux}[u,z]$
  and then the subpath~$P_{vy}[z,y]$. This path would mean
  that~$y \in N^\sig(u)$, contradicting our choice of~$y$.

  Hence, assume $P_{vy}$
  and $P_{ux}$ do not intersect. But then the graph $P_{ux} \cup P_{vx} \cup P_{vy}$ is
  connected and contains every colour of $\sig$ at least twice. Since
  $|\sig| \leq 2r+1$ this contradicts our assumption that the colouring~$c$
  is $(2r+2)$-centred.
\end{proof}

\noindent
We see that a single proper signature~$\sig$ imposes a very restricted
structure on the respective $\sig$-neighbourhoods in the graph. Even more interesting
is the interaction of proper signatures with each other. To that end, let us
introduce the notion of~\emph{$(X,\sig)$-equivalence}: vertices~$u$ and~$v$
are equivalent if their respective~$\sig$-neighbourhoods in~$X$ are the same, \ie
\[
  u \simeq^{X}_{\sig} v \iff  N^\sig(u) \cap X = N^\sig(v) \cap X.
\]

\begin{lemma}\label{lem:sigma-laminar}
  Let~$\sig_1,\sig_2$ be a pair of proper signatures. Let 
  further~$Y_{\sig_1,\sig_2} = N^{-\sig_1}(X) \cap N^{-\sig_2}(X)$ be 
  all vertices that can reach at least one vertex in~$X$ via a $\sig_1$-path
  and at least one vertex via a $\sig_2$-path.

  Fix two arbitrary equivalence classes~$C_{\sig_1} \in Y_{\sig_1,\sig_2} / {\simeq^X_{\sig_1}}$
  and~$C_{\sig_2} \in Y_{\sig_1,\sig_2} / {\simeq^X_{\sig_2}}$.
  Then either~$C_{\sig_1} \cap C_{\sig_2} = \emptyset$, $C_{\sig_1} \subseteq C_{\sig_2}$,
  or~$C_{\sig_1} \supseteq C_{\sig_2}$.
\end{lemma}
\begin{proof}
  The statement is trivial if $\sig_1 = \sig_2$ or $C_{\sig_1} = C_{\sig_2}$. Otherwise, assume
  that there exist $C_{\sig_1} \neq C_{\sig_2}$ such that indeed $C_{\sig_1}$ and
  $C_{\sig_2}$ are not related in the three above ways. Since this is
  impossible when $|C_{\sig_1}| = 1$ or $|C_{\sig_2}| = 1$, we know that
  there exists vertices $u,v,w \in Y_{\sig_1,\sig_2}$ 
  with $u \in C_{\sig_1} \setminus C_{\sig_2}$, $v \in C_{\sig_2}
  \setminus C_{\sig_1}$ and $w \in C_{\sig_1} \cap C_{\sig_2}$.

  The respective membership in these classes tell us the following
  about the vertices $u,v,w$: 
  \begin{align*}
    N^{\sig_1}(u) \cap X &= N^{\sig_1}(w) \cap X \neq N^{\sig_1}(v) \cap X \quad \text{and} \\
    N^{\sig_2}(u) \cap X &\neq N^{\sig_2}(w) \cap X = N^{\sig_2}(v) \cap X.
  \end{align*}
  Using Lemma~\ref{lem:sigma-nb} we can strengthen this
  statement: $N^{\sig_1}(u) \cap N^{\sig_1}(v) = \emptyset$ and $N^{\sig_2}(u)
  \cap N^{\sig_2}(v) = \emptyset$ and since $u,v,w$ are contained 
  in $Y_{\sig_1,\sig_2}$, we know that all the involved neighbourhoods intersect~$X$.

  Therefore, we can pick distinct vertices $x_1,y_1,x_2,y_2 \in X$ such that
  $x_1 \in N^{\sig_1}(u)$, $y_1 \in N^{\sig_1}(v)$ and $x_2 \in
  N^{\sig_2}(u)$, $y_2 \in N^{\sig_2}(v)$.

  \begin{center}
    \includegraphics[scale=\figscale]{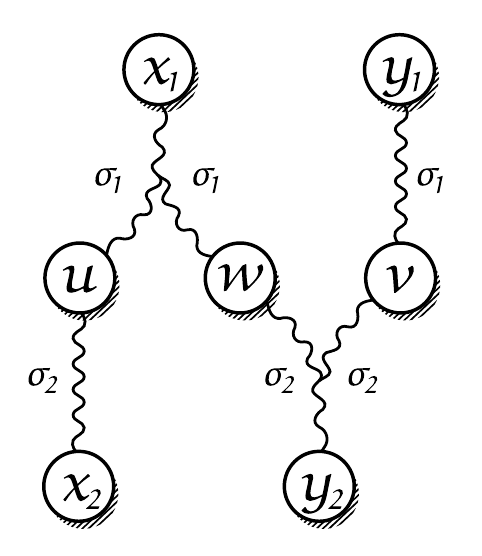}
  \end{center}

  \noindent
  Since $N^{\sig_1}(w) = N^{\sig_1}(u)$, we can connect the vertices $u,w$
  with two (not necessarily disjoint) $\sig_1$-paths $P^{\sig_1}_u,
  P^{\sig_1}_w$ that start both in $x_1$. Further, there exists a
  $\sig_1$-path $P^{\sig_1}_v$ from $y_1$ to $v$. If $P^{\sig_1}_v$ would
  intersect either $P^{\sig_1}_u$ or $P^{\sig_1}_w$, we could not have that
  $N^{\sig_1}(v) \cap N^{\sig_1}(u) =
  \emptyset$ according to Lemma~\ref{lem:sigma-nb}. We conclude that indeed
  $P^{\sig_1}_v$ is disjoint from both $P^{\sig_1}_u$ and $P^{\sig_1}_w$.

  We repeat the same construction for $x_2,y_2$ and the signature
  $\sig_2$ to obtain paths $P^{\sig_2}_u, P^{\sig_2}_v, P^{\sig_2}_w$. 
  This time, $P^{\sig_2}_u$ is necessarily disjoint from both $P^{\sig_2}_v$ and $P^{\sig_2}_w$
  (\cf figure above). 
  We reach a contradiction: observe that the graph induced by
  the paths 
  $P^{\sig_1}_u, P^{\sig_1}_v, P^{\sig_1}_w,
  P^{\sig_2}_u, P^{\sig_2}_v, P^{\sig_2}_w$ is connected, 
  contains every colour of $\sig_1 \cup \sig_2$ at least twice
  and in total at most $2r+1$ colours. This is impossible if
  $c$ was indeed $(2r+2)$-centred.
\end{proof}

\noindent
For the next lemma we extend the notion of~$(X,\sig)$-equivalence to
sets of proper signatures~$\Sig$. We define the
$(X,\Sig)$-equivalence relation on the vertices of $G$ as
follows:
\[
  u \simeq^{X}_{\Sig} v \iff \text{for all } \sig \in \Sig,
  N^\sig(u) \cap X = N^\sig(v) \cap X
\]

\begin{lemma}\label{lem:proper-sigma-complexity}
  Let~$\hat \Sig \subseteq \Sig_{\leq r}$ be a set of proper signatures and 
  let $W_{\hat \Sig} = \bigcap_{\sig \in \hat \Sig} N^{-\sig}(X)$ be those vertices
  in~$G$ which have a non-empty $\sig$-neighbourhood in~$X$ for every~$\sig \in \hat \Sig$.
  Then $|W_{\hat \Sig} / {\simeq^X_{\hat \Sig}}| \leq |\hat \Sig| \cdot |X|$. 
\end{lemma}
\begin{proof} 
  Define the set family~$\cal F := \bigcup_{\sig \in \hat \Sig}
  (W_{\hat \Sig} / {\simeq^X_{\sig}})$ of the classes of all
  equivalence relations defined via
  a signature contained in~$\hat \Sig$.
  By Lemma~\ref{lem:sigma-laminar} and our
  choice of~$W_{\hat \Sig}$, every pair 
  $B_1, B_2 \in \cal F$ satisfies $B_1 \cap B_2 \in \{ \emptyset, B_1, B_2 \}$
  (\ie $\cal F$ is a \emph{laminar} family). 

  Consider a class~$B \in  W_{\hat \Sig}/{\sigeq_{\hat \Sig}}$. Then~$B$ 
  is the result of a intersection of at most~$|\hat \Sig|$ classes in~$\cal F$.
  Since~$B \neq \emptyset$ and~$\cal F$ is laminar, it follows that 
  $B \in \cal F$. We conclude that 
  \begin{align*}
    |W_{\hat \Sig} / {\simeq^X_{\hat \Sig}}| \leq |\cal F| \leq |\hat \Sig| \cdot |X|,
  \end{align*}
  where the second inequality follows from Lemma~\ref{lem:sigma-nb}.\qedhere
\end{proof}

\noindent
In order to apply the above lemma it is left to bound the number of possible
$r$-neighbourhoods in~$X$ by $\sigma$-neighbourhoods of \emph{proper}
signatures. We establish this bound by successively refining the $(X,r)$-twin
equivalence. The following figure gives an overview over the proof (using
relations yet to be introduced).

\begin{comment}
\def\nudge{.4em}
\begin{align*}
  u &\nbeq_{r-1} v           &&\iff &
                              N^{r-1}[u] \cap X &= N^{r-1}[v] \cap X \\
  &\hspace{\nudge}\Big\Uparrow \small\mathrlap{
    ~\text{Lemma~\ref{lemma:struct-nb-equiv}}} \\
%  
  u &\snbeq_{r-1} v   &&\iff &
                              \big(N^i(u)\cap X \big)_{0 \leq i < r} 
                                    &= \big(N^i(v)\cap X \big)_{0 \leq i < r} \\
%                                    
  &\hspace{\nudge}\Big\Uparrow \small\mathrlap{
        ~\text{Lemma~\ref{lemma:sig-nb-equiv}}} \\
  u &\rsigeq_r v    &&\iff &
                              \big(N^\sig(u) \cap X \big)_{\sig \in \Sig_{\leq r}} 
                                    &= \big(N^\sig(v) \cap X \big)_{\sig \in \Sig_{\leq r}}\\
%                                    
  &\hspace{\nudge}\Big\Uparrow \small\mathrlap{~\text{Lemma~\ref{lemma:sig-prop-nb-equiv}}} \\
  u &\prsigeq_r v     &&\iff &
                               \big(N^{\sig}(u^1) \cap X^{|\sig|} \big)_{\sig \in \hat \Sig_{\leq r}} 
                                  &=  \big(N^{\sig}(v^1) \cap X^{|\sig|} \big)_{\sig \in \hat \Sig_{\leq r}} 
\end{align*}
\end{comment}

\def\nudge{.4em}
\begin{align*}
  u &\nbeq_{r-1} v           &&\iff &
                              N^{r-1}[u] \cap X &= N^{r-1}[v] \cap X \\
  &\hspace{\nudge}\Big\Uparrow \small\mathrlap{
        ~\text{Lemma~\ref{lemma:sig-nb-equiv}}} \\
  u &\rsigeq_r v    &&\iff &
                              \big(N^\sig(u) \cap X \big)_{\sig \in \Sig_{\leq r}} 
                                    &= \big(N^\sig(v) \cap X \big)_{\sig \in \Sig_{\leq r}}\\
  &\hspace{\nudge}\Big\Uparrow \small\mathrlap{~\text{Lemma~\ref{lemma:sig-prop-nb-equiv}}} \\
  u &\prsigeq_r v     &&\iff &
                               \big(N_{\hat G}^{\hat \sig}(u^1) \cap X^{|\hat \sig|} \big)_{\hat \sig \in \hat \Sig_{\leq r}} 
                                  &=  \big(N_{\hat G}^{\hat \sig}(v^1) \cap X^{|\hat \sig|} \big)_{\hat \sig \in \hat \Sig_{\leq r}} 
\end{align*}

\noindent
Where the last relation is defined with the help of an
auxiliary graph~$\hat G$ and signature set~$\hat \Sig_{\leq r}$ whose
construction is described later. The bound on the index of this last relation
will prove Theorem~\ref{thm:nb-bound-centred}.

\begin{lemma}\label{lemma:sig-nb-equiv}
  The equivalence relation~$\rsigeq_r$ over~$V(G)$ defined via
  \[
    u \rsigeq_r  v \iff
    \big(N^\sig(u) \cap X \big)_{\sig \in \Sig_{\leq r}} 
    = \big(N^\sig(v) \cap X \big)_{\sig \in \Sig_{\leq r}} 
  \] 
  is a refinement of~$\nbeq_{r-1}$.
\end{lemma}
\begin{proof}
  Assume $u \rsigeq_r v$. We need to prove that $N^{r-1}[u] \cap X = N^{r-1}[v] \cap X$.  The equivalence of~$u$ and~$v$ implies that
  % We want to align 
  %\settowidth{\templen}{\phantom{$\implies$}}  following
  \begin{align*}
  %  \big( N^\sig(u) \cap X \big)_{\sig \in \Sig_{=i}} &\makebox[\templen]{$=$}
   % \big( N^\sig(v) \cap X \big)_{\sig \in \Sig_{=i}}  
  %\intertext{
  %for every~$1 \leq i \leq r$. Therefore we have that
  %}
    w \in N^{r-1}[v] \cap X &\iff \exists \sig \in \Sig_{\leq r}\colon w \in N^{\sig}(v) \cap X \\
    &\iff \exists \sig \in \Sig_{\leq r}\colon w \in N^{\sig}(u) \cap X \\
    &\iff w \in N^{r-1}[u] \cap X.\qedhere
  \end{align*}
\end{proof}

\noindent
We now construct an auxiliary graph and colouring as follows: Let~$\hat G = G
\lexprod \comp K_r$. 
Assuming that~$V(K_r) = [r]$ and hence $V(\hat G) = V(G) \times [r]$, we will
use the shorthand~$v^i = (v,i)$ for $v \in V(G), i \in [r]$ and call~$v^i$ the
\emph{$i$th copy} of~$v$. Using this notation, we define a colouring $\hat c\colon V(\hat G) \to [\xi] \times [r]$
of~$\hat G$ via $\hat c(v^i) = (c(v),i)$.
Note that~$\hat c$ is a $(2r+2)$-centred colouring
of~$\hat G$: any connected subgraph~$\hat H \subseteq \hat G$ with less
than~$2r+2$ colours and no centre would directly imply that the subgraph~$H
\subseteq G$ with vertex set~$V(H) = \bigcup_{1 \leq i \leq r} \{ v \in V(G) \mid
v^i \in \hat H \}$ contains at most~$2r+2$ colours and no centre, contradicting
our choice of~$c$. 

For a signature~$\sig \in \Sig_{\leq r}$ we define the proper signature
$\hat \sig = ( (\sig[i],i))_{1\leq i \leq |\sig|}$. Accordingly,
we define the set of proper signatures~$\hat \Sig_{\leq r}$ over colours $[\xi]
\times [r]$ as $\hat \Sig_{\leq r} = \{\hat \sig \mid \sig \in \Sig_{\leq r}
\}$.
The following lemma
connects the sigma-equivalence~$\rsigeq_r$ over $V(G)$ with a suitable
equivalence defined over the above auxiliary structure.

\begin{lemma}\label{lemma:sig-prop-nb-equiv}
  The equivalence relation $\prsigeq_r$ over~$V(G)$ defined via
  \[
      u \prsigeq_r v   \iff 
              \big(N_{\hat G}^{\hat \sig}(u^1) \cap X^{|\hat \sig|} \big)_{\hat \sig \in \hat \Sig_{\leq r}} 
              =  \big(N_{\hat G}^{\hat \sig}(v^1) \cap X^{|\hat \sig|} \big)_{\hat \sig \in \hat \Sig_{\leq r}}
  \]
  is a refinement of~$\rsigeq_r$ where~$X^i := \{ v^i \mid v \in X\}$.
\end{lemma}
\begin{proof}
  Assume $u \prsigeq_r v$. Then for every signature~$\hat \sig \in \Sig_{\leq r}$
  we have that $N_{\hat G}^{\hat \sig}(u^1) \cap X^{|\hat \sig|} = N_{\hat G}^{\hat \sig}(v^1)
  \cap X^{|\hat \sig|}$. Now note that if~$w^{|\hat \sig|}$ is~$\hat \sig$-reachable
  from~$u^1$ in~$\hat G$, then~$w$ is~$\sig$-reachable from~$u$ in~$G$:
  if~$u^1 x_2^2 \ldots x_{{|\hat \sig|-1}}^{|\hat \sig|-1} w^{|\hat \sig|}$ is a $\hat \sig$-path
  in~$\hat G$, then~$u x_1 \ldots x_{|\hat \sig|-1} w$ is, by construction of~$\hat \sig$,
  a $\sig$-path in~$G$.

  Accordingly
  $w^{|\hat \sig|} \in N^{\hat \sig}(u^1)$ implies that~$w \in N^{\sig}(u)$.
  We conclude that therefore $N^{\sig}(u) \cap X^{|\sig|} = N^{\sig}(v) \cap X^{|\sig|}$
  and thus $u \rsigeq_r v$.
\end{proof}

\begin{lemma}\label{lemma:sig-bound}
  $| V(G) / {\prsigeq_r} | \leq   r2^{\xi^{r+1}} \cdot |X|.$
\end{lemma}
\begin{proof}
  To obtain the bound, we apply Lemma~\ref{lem:proper-sigma-complexity} to
  every subset of signatures~$\hat \Sig \subseteq \hat \Sig_{\leq r}$. 
  Let~$\hat X \subseteq \hat G$ be the set containing all copies of vertices in~$X$. 
   \begin{align*}
    | V(G) / {\prsigeq_r} | &\leq |V(\hat G)/ \simeq^{\hat X^r}_{\hat \Sig_{\leq r}} |
    \leq \sum_{\hat \Sig \subseteq \hat \Sig_{\leq r} } |\hat \Sig| \cdot |\hat X| 
    = r2^{\xi^{r+1}} \cdot |X| \qedhere
  \end{align*}  
\end{proof}

\noindent
The proof of this section's theorem is now only a technicality.

\begin{proof}[Proof of Theorem~\ref{thm:nb-bound-centred}]
  By Lemma~\ref{lemma:sig-nb-equiv} and~\ref{lemma:sig-prop-nb-equiv} we have
  that
  \begin{gather*}
    |V(G) / {\nbeq_{r}}|  \leq |V(G) / {\rsigeq_{r+1}}|
    \leq |V(G) / {\prsigeq_{r+1}}|
  \end{gather*}
  Which, by Lemma~\ref{lemma:sig-bound}, is 
  at most~$(r+1)2^{{\chi_{2r+2}(G)}^{r+2}} \cdot |X|$ and the claim follows.
\end{proof}

%  dP   dP   dP  a88888b.  .88888.  dP        
%  88   88   88 d8'   `88 d8'   `8b 88        
%  88  .8P  .8P 88        88     88 88        
%  88  d8'  d8' 88        88     88 88        
%  88.d8P8.d8P  Y8.   .88 Y8.   .8P 88        
%  8888' Y88'    Y88888P'  `8888P'  88888888P 
%                                             
%    
\section{Neighbourhood Complexity and Weak Colouring Number}\label{sec:wcol}

\noindent
Having obtained a bound for the neighbourhood complexity in terms
of the $r$-centred colouring number, we now derive a bound in terms of the weak $r$-colouring number.
For the next proof, we say that two vertices $u,v \in V(G)$ have \emph{the same distances} 
to $Z\subseteq V(G)$ if for every $z\in Z$ we have $d_G(u,z)=d_G(v,z)$.

\nbweak*
% \begin{theorem}
%   For every graph~$G$ it holds that
%   \[
%     \nu_r(G) \leq 2^{\wcol_{2r}(G)}\wcol_{2r}(G).
%   \]
% \end{theorem}
\begin{proof}

    Fix a graph~$G$ and choose any subset~$\emptyset\neq X \subseteq V(G)$. We will show in the following that 
  \[
  |V(G) / {\nbeq_r}| \leq \left(\frac{1}{2}(2r+2)^{\wcol_{2r}(G)} \wcol_{2r}(G) + 1\right) |X|,
  \]
  from which the claim immediately follows.
  
  Let $\alpha_0 \in V(G) /{\nbeq_r}$ be the equivalence class of ${\nbeq_r}$ 
  corresponding to the vertices of $G$ with an empty $r$-neighbourhood in $X$
  and let $\mathcal W=\left(V(G) /{\nbeq_r}\right)\setminus \{\alpha_0\}$. Moreover, let $L\in\Pi(G)$ be 
  such that $\wcol_{2r}(G)=\max_{v\in  V(G)}|\Wreach_{2r}[G,L, v]|$. 
  We will estimate the neighbourhood complexity of $X$ via the neighbourhood complexity of a
  certain \emph{good} subset of $\Wreach_{r}[G,L, X].$

  For a vertex $v \in N^r(X)$ and a vertex $x \in N^r[v]\cap X$, let $\mathcal{P}_v^x$ 
  be the set of all shortest $(v,x)$-paths (of length at most $r$). 
  We define as $G^r[v]$ the graph induced by the union of the paths of all $\mathcal{P}_v^x$,
  namely
  \[
    G^r[v] = G\Big[\bigcup_{x \in N^r[v] \cap X} \bigcup_{P \in \mathcal P^x_v} V(P)\Big]. %G\big[  \bigcup_{x\in N^r(v)\cap X} V(P_v^x)  \big].
  \]
  By its construction,~$G^r[v]$ contains, for every~$x \in N^r[v] \cap X$, all shortest paths
  of length at most~$r$ that connect~$v$ to~$x$.

  Now, for every equivalence class $\kappa \in \mathcal W$, choose a representative
  vertex $v_{\kappa} \in \kappa$. Let $C = \{ v_{\kappa} \}_{\kappa \in \mathcal W}$ be
  the set of representative vertices for all classes in $\mathcal W$. Using the representatives from $C$,
  we define for every class~$\kappa\in \mathcal W$ the set (see Fig.~\ref{fig:wcol-proof})
  \[
    Y_{\kappa}=\Wreach_{r}[G^r[v_\kappa],L,v_\kappa]  
      \cap  \Wreach_{r}[G,L,  N^r[v_\kappa] \cap X]
  \]
  and join all such sets into
  $
    Y = \bigcup_{\kappa \in \mathcal W} Y_\kappa.
  $
  Then, 
  \[
      Y\subseteq \bigcup_{\kappa \in \mathcal W} \Wreach_{r}[G,L,  N^r[v_\kappa] \cap X]\subseteq \Wreach_r [G, L, X].
  \]
  Moreover, by definition and the fact that $L$ is an ordering achieving $\wcol_{2r}(G)$ (and not necessarily one achieving $\wcol_{r}(G)$), we have
  \[
    |Y_{\kappa}|\leq |\Wreach_r [G, L, v_\kappa]| 
	     \leq |\Wreach_{2r} [G, L, v_\kappa]|\leq \wcol_{2r}(G).
  \]
  Notice that for every $x \in N^r[v]\cap X$, the minimum vertex (according to~$L$) of a path in $\mathcal{P}_v^x$ will
  always belong to $Y_\kappa$, therefore the set~$Y_\kappa$ intersects every path of the sets 
  $\mathcal{P}_{v_\kappa}^x$ forming $G^r[v_\kappa]$.
  We want to see how many different equivalence classes of $\mathcal W$ produce the same $Y_\kappa$ set. 
  This will allow us to bound the 
  neighbourhood complexity of $X$ by relating it to the number of different $Y_\kappa$'s. \looseness-1
  
  Suppose that $\kappa\neq \lambda$ with $Y_\kappa=Y_\lambda=Z$. Recall that $Y_\kappa$ intersects all 
  the shortest paths from $v_\kappa$ to the vertices of $N^r[v_\kappa]\cap X$ and that $G^r[v_\kappa]$
  is formed by all such shortest paths. 
  Hence, if $v_\kappa$ and $v_\lambda$ have the same distances to $Z$,
  then we clearly get $N^r[v_\kappa]\cap X=N^r[v_\lambda]\cap X$, a contradiction. This means that if $Y_\kappa=Y_\lambda=Z$,
  the vertices $v_\kappa$ and $v_\lambda$ cannot have the same distances to $Z$. But there 
  are at most $(r+1)^{|Z|}$ possible configurations of distances of the vertices of a set $Z$ 
  to a vertex $v$ that has distance at most $r$ to every vertex of $Z$.  
  It follows that the number of equivalence classes of $\mathcal W$ that produce the same set~$Y_\kappa$
  through their representative $v_\kappa$ from $C$ is at most $(r+1)^{|Y_\kappa|}\leq (r+1)^{\wcol_{2r}(G)}$.
  
  \begin{figure}
  \begin{center}
    \includegraphics[scale=\smallfigscale]{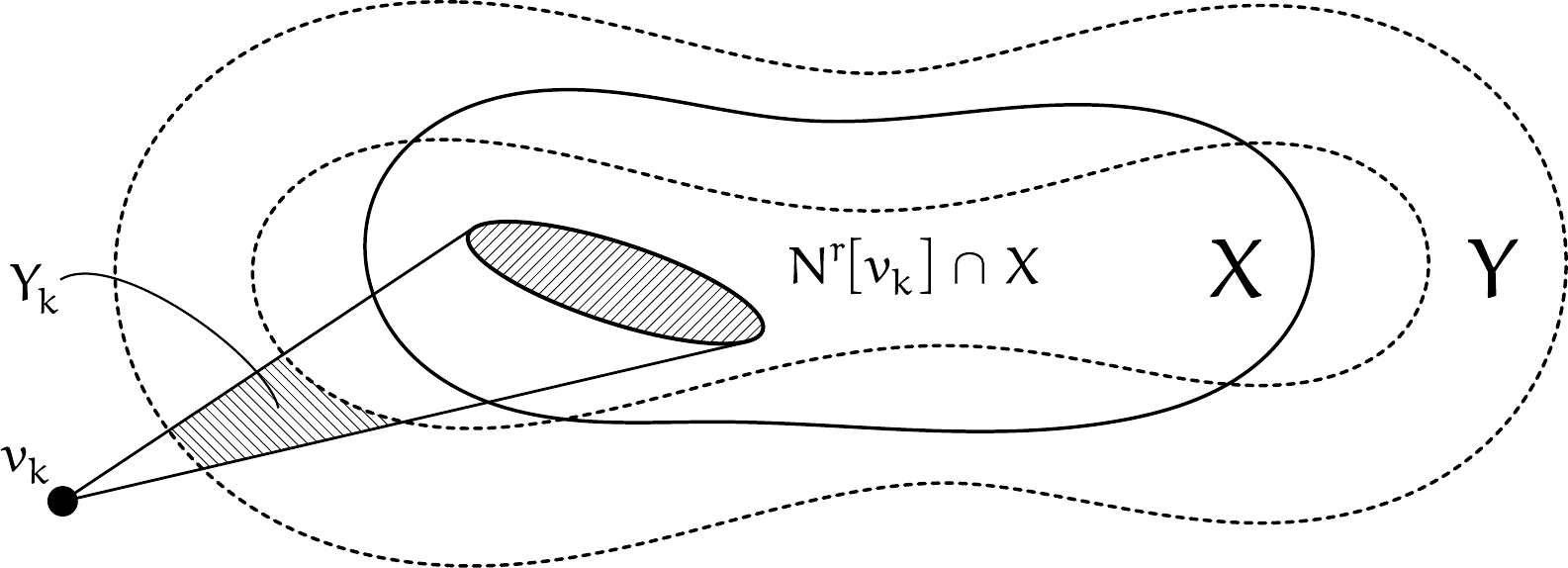}
  \end{center}
  \vspace*{-.5cm}
  \caption{\label{fig:wcol-proof}%
    A set $Y_\kappa$ and the set $Y$. %Illustration for the proof of Theorem~\ref{thm:nb-bound-weak}.
  }
\end{figure}

  Let $\mathcal{Y}:=\{Y_\kappa \mid \kappa \in \mathcal W\}$ be the set of all
   (\emph{different}) $Y_\kappa$'s, and define $\gamma \colon \mathcal{Y}\rightarrow Y$ by
  $\gamma(Y_\kappa)=\argmax_{y\in Y_\kappa } L(y)$. That is, $\gamma(Y_\kappa)$ is
  that vertex in~$Y_\kappa$ that comes last according to~$L$. Observe that---by definition---every vertex in~$Y_\kappa$ is weakly $r$-reachable
  from~$v_\kappa$. It follows that
  every vertex in~$Y_\kappa$ is weakly $2r$-reachable from~$\gamma(Y_\kappa)$ via $v_\kappa$.
  In other words, $Y_\kappa\subseteq \Wreach_{2r}[G,L,\gamma(Y_\kappa)]$.
  Consequently, for every vertex~$y\in \gamma(\mathcal{Y})$, it holds that\footnote{We remind the reader that this union expresses the union of a set in the set theoretical sense,
    i.e.~the union of a set is the union of all of its elements (as sets).}
  \[
    \bigcup \gamma^{-1}(y) \subseteq \Wreach_{2r}[G,L,y],
  \]
  i.e. the union $\bigcup \gamma^{-1}(y)$ of all $Y_\kappa$'s that choose the same vertex~$y$ via $\gamma$ has
  size at most~$\wcol_{2r}(G)$. But every set in the family $\gamma^{-1} (y)$ is a subset of $\bigcup \gamma^{-1}(y)$ that contains $y$. Since there are at most
$2^{|\bigcup \gamma^{-1}(y )|-1}$ different such subsets of $\bigcup \gamma^{-1}(y )$, the number of \emph{different} $Y_\kappa$'s for which the same vertex
  is chosen via $\gamma$ is bounded by~$2^{\wcol_{2r}(G)-1}$, i.e.
  \[
   |\gamma^{-1}(y)| \leq 2^{\wcol_{2r}(G)-1}.
  \]
  Recalling that one $Y_\kappa$ corresponds to at most
  $(r+1)^{\wcol_{2r}(G)}$ equivalence classes  of $\mathcal W$ and that $Y\subseteq \Wreach_r [G, L, X]$, we can now bound the size of $\mathcal W$ as follows:
  \begin{align*}
  |\mathcal W|
      &\leq (r+1)^{\wcol_{2r}(G)}\cdot |\mathcal{Y}|= (r+1)^{\wcol_{2r}(G)}\cdot \sum_{y\in \gamma(\mathcal{Y})}|\gamma^{-1}(y)| \\
  %   &\leq (r+1)^{\wcol_{2r}(G)} \cdot \sum_{y\in \gamma(\mathcal{Y})}2^{\left|\bigcup{\gamma^{-1}(y)} \right|} \\
      &\leq (r+1)^{\wcol_{2r}(G)} \cdot \sum_{y\in \gamma(\mathcal{Y})}2^{\wcol_{2r}(G)-1} \\
      &= \frac{1}{2}(2r+2)^{\wcol_{2r}(G)} \cdot   |\gamma(\mathcal{Y})|
  \end{align*}
  from which we obtain that
  \begin{align*}
   |V(G) / {\nbeq_r}| &\leq |\mathcal W| +1 \leq \frac{1}{2}(2r+2)^{\wcol_{2r}(G)} \cdot   |\gamma(\mathcal{Y})|+1 \\
                      &\leq \frac{1}{2}(2r+2)^{\wcol_{2r}(G)} \cdot |Y| +1\\
                      &\leq \frac{1}{2}(2r+2)^{\wcol_{2r}(G)} \cdot |\Wreach_{r}[G,L, X]|+1 \\
                      &\leq \frac{1}{2}(2r+2)^{\wcol_{2r}(G)}\wcol_{2r}(G) \cdot |X|+1 \\
                      &\leq \left(\frac{1}{2}(2r+2)^{\wcol_{2r}(G)}\wcol_{2r}(G)+1\right)|X|,
  \end{align*}
  as claimed.
\end{proof}

%   a88888b. dP     dP   .d888888   888888ba  
%  d8'   `88 88     88  d8'    88   88    `8b 
%  88        88aaaaa88a 88aaaaa88a a88aaaa8P' 
%  88        88     88  88     88   88   `8b. 
%  Y8.   .88 88     88  88     88   88     88 
%   Y88888P' dP     dP  88     88   dP     dP 
%                                             
%         
\section{Completing the Characterisation}

\noindent
We have seen in the previous two sections that bounded expansion implies
bounded neighbourhood complexity. Let us now prove the other direction to
arrive at the full characterisation. We begin by proving that every bipartite
graph with low neighbourhood complexity must have low minimum degree. To that
end, we will need the following Lemma.

% Lemma 4.4 in Sparsity (p. 74)
\begin{lemma}[\Nesetril \& Ossona de Mendez~\cite{Sparsity}]\label{lemma:skewed-bipartite}
  Let~$G = (A,B,E)$ be a bipartite graph and let~$1 \leq r \leq s \leq |A|$.
  Assume each vertex in~$B$ has degree at least~$r$. 

  Then there exists a subset~$A' \subseteq A$ and a
  subset~$B' \subseteq B$ such that $|A'| = s$ and~$|B'| \geq |B|/2$
  and every vertex in~$B'$ has at least $r\frac{|A'|}{|A|}$ neighbours
  in~$A'$.
\end{lemma}

\noindent
The minimum degree and depth-one neighbourhood complexity $\nu_1$ of a
bipartite graph can now be related to each other as follows:

\begin{comment}
\begin{lemma}\label{lemma:mindeg-nu-0}
  Let~$G = (A,B,E)$ be a bipartite graph. Then~$\delta(G) \leq \nu_0(G)^3$.
\end{lemma}
\begin{proof}
  \def\r{4\nu^3 \log \nu}
  \def\t{2\nu^3}
  \def\thalf{\nu^3}
  \def\rovert{2\log \nu}
  Assume without loss of generality that~$|B| \geq |A|$. Let~$\nu = \nu_0(G)$
  and assume that~$\delta(G) > \nu^3$.
  We apply Lemma~\ref{lemma:skewed-bipartite} with~$r = \r$ and~$s = |A|/\t$
  and obtain a subgraph~$G'=(A',B',E')$ with
  \begin{itemize}
    \item $|A'| = |A|/\t$,
    \item $|B'| = |B|/2$ and thus~$|B'| \geq |A'| \thalf$,
    \item and~$\deg_{G'}(v) \geq \r \cdot \frac{|A'|}{|A|} = \r / \t = \rovert. $ for~$v \in B'$.
  \end{itemize}
  Note that if~$K_{\nu^2,2\log \nu}$ is a subgraph of~$G'$ then
  \[
    \nu_0(G) \geq \nu_0(G') \geq \frac{\nu^2}{2\log \nu} > \nu,
  \]
  a contradiction. Let us partition~$B'$ into
  twin-classes~$B'_1, \ldots B'_\ell$. Since each twin-class has at
  least~$2 \log \nu$ neighbours, the size of each twin class must be
  bounded by~$|B_i| < \nu^2$ and therefore the number of twin-classes
  is at least~$\ell > |B'| / \nu^2$. Since each twin-class has, by
  definition, a unique neighbourhood in~$A'$, we conclude that
  \[
    \nu_0(G') \geq \frac{\ell}{|A'|} > \frac{|B'|}{\nu^2} \frac{\nu^3}{|B'|} = \nu,
  \]
  a contradiction.
\end{proof}
\end{comment}

\begin{lemma}\label{lemma:mindeg-nu-1}
  Let~$G = (A,B,E)$ be a non-empty bipartite graph.
  Then 
  \[
    \delta(G) < 4\nu_1(G)\big(2\ceil{\log\nu_1(G)}+1\big)\big(64\nu_1(G)^3 \ceil{\log\nu_1(G)}+16\nu_1(G)^2+1\big).
  \]
\end{lemma}
\begin{proof}
  \def\r{8\nu^3\log\nu+4\nu^2}
  \def\t{\nu(2\log\nu+1)}
  \def\thalf{\nu^3}
  \def\rovert{2\nu^2}
  Let 
  \[
    \alpha = 4\nu_1(G)\big(2\ceil{\log\nu_1(G)}+1\big)\big(64\nu_1(G)^3 \ceil{\log\nu_1(G)}+16\nu_1(G)^2+1\big)
  \]
  and suppose that $\delta(G) \geq \alpha.$ Assume without loss of generality
  that~$|B| \geq |A|$ and let~$\nu = 2^{\ceil{\log\nu_1(G)}}$. Observe
  that both $\nu,\log\nu$ are integers and that $\nu_1(G)\leq\nu <2\nu_1(G)$.
  Therefore, 
  \[
    |B| \geq |A|\geq \delta(G) > 2\t\big(\r+1\big).
  \]
  Let us apply Lemma~\ref{lemma:skewed-bipartite} on $G$ with~$r = \r+1$ and~$s =
  \lfloor\frac{|A|}{2\t}\rfloor$. Notice that this is indeed possible, because
  $|A|>2\t\cdot r$ and therefore $s\geq r$. We obtain a
  subgraph~$G'=(A',B',E')$ with
  \begin{enumerate}
    \item $ \frac{|A|}{2\t} -1 < |A'| =s \leq \frac{|A|}{2\t}$,
    \item $|B'| \geq \frac{|B|}{2}$, and thus~$|B'| \geq \frac{|A|}{2} \geq \t|A'|$,
    \item
      and such that for every~$v \in B'$ we have that $\deg_{G'}(v) \geq r \cdot \frac{|A'|}{|A|}$.
  \end{enumerate}
  Combining the first and third property with~$|A|>2\t\cdot r$, we obtain
  \begin{align*}
    \deg_{G'}(v) &\geq r \cdot \frac{|A'|}{|A|} > r \Big(\frac{1}{2\t}-\frac{1}{|A|}\Big) \\
                 &> r \Big(\frac{1}{2\t}-\frac{1}{2\t\cdot r}\Big) \\
                 &= \frac{r-1}{2\t}= \rovert.
  \end{align*}  
  Now, note that any graph $H$ with at least two vertices trivially has
  $\nu_1(H)\geq 2$ by taking $X$ to be a single vertex of $H$. Hence,
  if~$K_{\rovert,2\log \nu+1}$ is a subgraph of~$G'$, we have that
  \[
    \nu_1(G) \geq \nu_1(G') \geq \nu_1(K_{\rovert,2\log \nu+1}) \geq
    \frac{\rovert}{2\log \nu+1} > \nu,
  \]
  where the last inequality follows by the fact that $\nu\geq 2$, a contradiction. 
  
  So, let us call two vertices $u,v\in V(G')$ \emph{twins} if $N^1_{G'}(u)=N^1_{G'}(v)$ and let us partition~$B'$ into
  twin-classes~$B'_1, \ldots B'_\ell$. Since each twin-class has at
  least~$\rovert$ neighbours, the size of each twin-class must be
  bounded by~$|B'_i| < 2\log\nu+1$. Hence, the number of
  twin-classes is at least~$\ell > \frac{|B'|}{2\log\nu+1}$. Since each
  twin-class has, by definition, a unique neighbourhood in~$A'$, we
  conclude that
  \[
    \nu_1(G') \geq \frac{\ell}{|A'|} > \frac{|B'|}{2\log\nu+1}
    \frac{\t}{|B'|} = \nu\geq \nu_1(G),
  \]
  a contradiction.
\end{proof}

\noindent
It easily follows that every graph with low neighbourhood complexity must have low average degree.

\begin{corollary}\label{cor:nu-subgraph}
  Let~$G$ be a graph. Then $\topgrad_0(G) < 5445 \cdot \nu_1(G)^4 \log^2 \nu_1(G)$.
\end{corollary}
\begin{proof}
  We assume that~$\topgrad_0(G) = \|G\|/|G|$, otherwise we restrict ourselves
  to a suitable subgraph of~$G$ with that property. The case where $|G|=1$ is trivial, therefore we may assume that $|G|\geq 2$. 
  It is folklore that~$G$
  contains a bipartite graph~$H$ such that~$\|H\| \geq \|G\|/2$. We can further
  ensure that~$\delta(H) \geq \|H\|/|H|$ by excluding vertices of lower degree
  (this operation cannot decrease the density of~$H$). Applying Lemma~\ref{lemma:mindeg-nu-1}
  to~$H$, we obtain that
  \[
     \topgrad_0(G) = \frac{\|G\|}{|G|} \leq 2\frac{\|H\|}{|H|} \leq 2\delta(H).
  \]
  We apply the bound provided by Lemma~\ref{lemma:mindeg-nu-1} and relax it to
  the more concise polynomial~$(5445/2) \cdot \nu_1(G)^4 \log^2 \nu_1(G)$,
  using the fact that~$\nu_1(G) \geq 2$.
\end{proof}

\noindent
The next theorem now leads to the full characterisation as stated in
Theorem~\ref{thm:bndexp-equals-bndnc}.

\begin{theorem}\label{thm:bndnc-to-bndexp}
  For every graph~$G$ and every half-integer~$r$ it holds that
  \[
    \topgrad_r(G) \leq (2r+1) \max\big\{ 5445 \nu_1(G)^4 \log^2 \nu_1(G), \, \nu_2(G), \ldots, \, \nu_{\ceil{r+1/2}}(G) \big\}.
  \]  
\end{theorem}
\begin{proof}
  Fix~$r$ and let~$H \stminor^r G$ be an $r$-shallow topological minor of
  maximal density, \ie $\topgrad_0(H) = \topgrad_r(G)$. Let further~$\phi_V,\phi_E$
  be a topological minor embedding of~$H$ into~$G$ of depth~$r$.

  Let us label the edges of~$H$ by the respective path-length in the embedding
  $\phi_V,\phi_E$: an edge~$uv \in H$ receives the label~$\|\phi_E(uv)\|$.
  Let~$r'$ be the label of highest frequency and let~$H' \subgraph H$ be 
  the graph obtained from~$H$ by only keeping edges labelled with~$r'$.
  Since there were up to~$2r+1$ labels in~$H$, we have that~$(2r+1)\|H'\| \geq |H|$
  and therefore
  \begin{equation}\label{eq:grads}
    \topgrad_r(G) = \topgrad_0(H) \leq (2r+1)\frac{\|H'\|}{|H'|} \leq (2r+1)\topgrad_0(H').
  \end{equation}

  \noindent
  First, consider the case that~$r' = 1$, \ie $H'$ is a subgraph of~$G$. Combining 
  (\ref{eq:grads}) with Corrollary~\ref{cor:nu-subgraph}, we obtain
  \begin{align*}
    \topgrad_r(G) &\leq (2r+1)\topgrad_0(H') \leq (2r+1)\topgrad_0(G) \\
                  &\leq (2r+1) \cdot 5445\,\nu_1(G)^4 \log^2 \nu_1(G).
  \end{align*}
  Otherwise, assume that~$r' \geq 2$, \ie every edge of~$H'$ is embedded into a
  path of length at least~$2$ in~$G$ by~$\phi_V,\phi_E$. Construct the subgraph~$G' \subgraph G$
  that contains all edges and vertices involved in the embedding of~$H'$ into~$G$, that
  is, $G'$ has vertices~$\bigcup_{v \in H'} V(\phi_V(v)) \cup \bigcup_{e \in H'} V(\phi_E(e))$ and
  edges~$\bigcup_{e \in H'} E(\phi_E(e))$. 

  Let~$X = \bigcup_{v \in H'} V(\phi_V(v))$ and let~$S \subseteq V(G')$ be a
  set constructed as follows: for every edge~$e \in H'$ we add the middle
  vertex of the path~$\phi_E(e)$ to~$S$---in case $r'$ is odd, we pick one of
  the two vertices that lie in the middle of~$\phi_E(e)$ arbitrarily. Because~$X$
  is an independent set in~$G'$ and~$r' > 1$, every vertex in~$S$ has exactly two neighbours
  in~$X$ at distance~$\ceil{r'/2}$ in the graph~$G'$. By construction, there is a one-to-one
  correspondence between these $\ceil{r'/2}$-neighbourhoods and the edges of~$H'$. Accordingly,
  \[
    \|H'\| = | \{N_{G'}^{\ceil{r'/2}}(v) \cap X \}_{v \in S}| 
  \]
  and therefore, using also the fact that~$G'$ is a subgraph of~$G$,
  \[
    \frac{\|H'\|}{|H'|} = \frac{| \{N_{G'}^{\ceil{r'/2}}(v) \cap X \}_{v \in S}|}{|X|} \leq \nu_{\ceil{r'/2}}(G') \leq \nu_{\ceil{r'/2}}(G).
  \]
  which, taken together with (\ref{eq:grads}) and the fact that~$G'$ is a subgraph of~$G$, yields
  \[
    \topgrad_r(G) \leq (2r+1)\topgrad_0(H') \leq (2r+1) \nu_{\ceil{r'/2}}(G).
  \]
  Putting everything together, we finally arrive at
  \[
    \topgrad_r(G) \leq (2r+1) \max\big\{ 5445\,\nu_1(G)^4 \log^2 \nu_1(G), \, \nu_2(G), \ldots, \, \nu_{\ceil{r+1/2}}(G) \big\}.
  \]  
  proving the theorem.
\end{proof}

\noindent
We conclude that graph classes with bounded neighbourhood complexity have
bounded expansion. Theorem~\ref{thm:bndexp-equals-bndnc} follows by
Theorems~\ref{thm:nb-bound-centred}, \ref{thm:nb-bound-weak} 
and~\ref{thm:bndnc-to-bndexp}.

\section{Concluding Remarks}
\noindent
One should note that in Theorems~\ref{thm:nb-bound-centred} 
and~\ref{thm:nb-bound-weak} the derived bounds are \emph{exponential} in the
measures~$\chi_{2r+2}$ and~$\wcol_{2r}$. Consequently, we cannot use
neighbourhood complexity to characterise nowhere dense classes: in these
classes, the quantities~$\chi_r$ and~$\wcol_r$ can only be bounded by
$O(|G|^{o(1)})$ which only results in superpolynomial bounds for~$\nu_r$.

This constitutes an unusual phenomenon in the following sense: so far, every known characterisation of
bounded expansion translated to a direct characterisation of nowhere denseness, but this has not yet been the case
for neighbourhood complexity. It would be remarkable if one could only characterise the property of bounded expansion
through neighbourhood complexity and not that of nowhere denseness.
So far, it is only known that~$\nu_1$ is bounded 
by~$O(|G|^{o(1)})$ in nowhere dense classes~\cite{BndExpKernels}.
We pose as an interesting open question whether this holds true for
$\nu_r$ for all~$r$, or whether nowhere dense classes can indeed have 
a neighbourhood complexity that cannot be bounded by such a function.

\bibliographystyle{habbrv}
\bibliography{biblio,conf}

\end{document}